\documentclass[journal]{new-aiaa}
\usepackage[utf8]{inputenc}
\usepackage{textcomp}

\usepackage{amsmath,amsfonts,amssymb,amsthm}
\usepackage{bbm}

\usepackage[version=4]{mhchem}
\usepackage{siunitx}

\usepackage{longtable,tabularx}
\usepackage{booktabs}
\usepackage{makecell}
\usepackage{multirow}
\usepackage{hhline}
\usepackage{diagbox}
\usepackage{tabstackengine}
\usepackage{tablefootnote}

\usepackage{graphicx}
\usepackage{subcaption}
\usepackage{xcolor}

\usepackage{algorithm}
\usepackage{algpseudocode}

\usepackage[english]{babel}
\usepackage{bm}
\usepackage{etoolbox}
\usepackage{mathtools}
\usepackage{nth}
\usepackage{float}
\usepackage{dsfont}
\usepackage{soul}
\setstcolor{red}

\patchcmd{\thebibliography}
  {\settowidth}
  {\setlength{\parsep}{0pt}\setlength{\itemsep}{0pt plus 0.1pt}\settowidth}
  {}{}

\newtheorem{lemma}{Lemma}
\newtheorem{theorem}{Theorem}

\renewcommand{\t}{^{\mbox{\tiny\sf T}}} 

\newcommand{\R}{\mathbb{R}}
\newcommand{\N}{\mathcal{N}}

\renewcommand{\d}{\mathrm{d}}
\newcommand{\E}{\mathbb{E}}
\newcommand{\tr}{\text{tr}}

\allowdisplaybreaks

\title{\LARGE \bf Nonlinear Stochastic Density Steering via Gaussian Mixture Schr\"odinger Bridges and Multiple Linearizations}

\author{Mattia Mosso$^{1}$, George Rapakoulias$^{1}$, Yue Guan$^{2}$, Panagiotis Tsiotras$^{3}$}

\affil{%
\vspace{5mm}
$^{1}$ Graduate Student, School of Aerospace Engineering, Georgia Institute of Technology, Atlanta, GA, USA\\
$^{2}$ Postdoctoral Fellow, School of Aerospace Engineering, Georgia Institute of Technology, Atlanta, GA, USA\\
$^{3}$ David and Andrew Lewis Chair and Professor, School of Aerospace Engineering, and Institute for Robotics and Intelligent Machines, Georgia Institute of Technology, Atlanta, GA, USA}

\begin{document}

\maketitle



\section{Introduction}

Controlling the stochastic evolution of a dynamical system between prescribed probability distributions is a fundamental problem at the intersection of optimal control, probability theory, and statistical mechanics \cite{ridderhof2020chance}. A rigorous treatment of this problem forms the foundation of \emph{Stochastic Optimal Control} (SOC).
In the sequel, $x_t \!\in\! \mathbb{R}^n$ and $u_t \!\in\! \mathbb{R}^m$ denote the state and control input at time $t$, respectively. 
Let $w_t \!\in\! \mathbb{R}^p$ denote a $p$-dimensional Brownian motion process.
The SOC problem is briefly summarized as
\begin{subequations} \label{SOC}
\begin{align}
\inf_{u} \ &V(t,x) = \mathbb{E} \left[ \int_0^T L(s, x_s, u_s) \d s + h(x_T) \right], \\
&\d x_t = f_t(x_t) \, \d t + B_t u_t \, \d t + D_t \d w_t, 
\label{eq:Ito}
\end{align}
\end{subequations}
where $L$ denotes the running cost, $h$ is the terminal cost shaping the state distribution at the final time $T$, and $V(t,x)$ represents the expected cost-to-go. The state dynamics in \eqref{eq:Ito} are modeled as an Itô stochastic differential equation.
Solving \eqref{SOC} amounts to solving the Hamilton–Jacobi–Bellman (HJB) equation, which, for general nonlinear SOC problems, is analytically intractable, while numerical methods are cumbersome due to the curse of dimensionality \cite{exarchos2018stochastic}.

The time evolution of the state distribution $\rho_t$ governed by the controlled It\^{o} SDE~\eqref{eq:Ito}
is described by the Fokker-Planck-Kolmogorov (FPK) equation \cite{sarkka2019applied},
\begin{equation}
    \frac{\partial \rho_t}{\partial t}
    + \nabla \cdot \bigl[\rho_t \bigl(f_t + B_t u_t\bigr)\bigr]
    - \frac{1}{2}\operatorname{tr}\bigl(D_t D_t^{\top}
    \nabla^{2}\rho_t\bigr) \! = \! 0,
    \label{eq:FPK}
\end{equation}
where $\nabla^2 \rho_t$ denotes the Hessian of the probability density.
In our analysis, we use the FPK equation to show that a control law $u_t$ and a candidate density $\rho_t$ are consistent, that is, if $u_t$ is applied to the SDE \eqref{eq:Ito}, the resulting density will be $\rho_t$.

In many control problems, rather than implicitly controlling the final-state distribution via the terminal cost $h(x_T)$, an explicit terminal distribution constraint is imposed. 
This problem, which we refer to as a \emph{density control problem}, has been studied extensively in the context of uncertainty control in autonomous systems~\cite{chen2021controlling, caluya2021wasserstein, pilipovsky2024dust, saravanos2023distributed}, Schr\"odinger Bridge and optimal transport framework~\cite{chen2021stochastic, rapakoulias2024go, chen2016optimal, mei2024flow}, and mean-field control theory~\cite{liu2022deep, chen2018steering, rapakoulias2025steering}. 
While numerical solutions are expensive to obtain in a general setting, in the specific case where the boundary distributions are Gaussian, and the system dynamics are linear, efficient solutions can be obtained using semidefinite programming (SDP) through the \emph{Optimal Covariance Steering} (OCS) framework \cite{chen2015optimal2, bakolas2016optimal, liu2022optimal, rapakoulias2023discrete} or even in closed form \cite{chen2015optimal, ito2024maximum}.

While significant progress has been made in the field of covariance and density steering applied to complex aerospace stochastic planning problems~\cite{ridderhof2019nonlinear, benedikter_convex_2022, kumagai2025robust, kumagai2025hands}, existing methods are largely restricted to linear (or locally linearized) dynamics.
A major limitation emerges in regimes of large uncertainty.
To this end, define the first-order Taylor expansion of the dynamics around the mean $\mu(t) = \mathbb{E}[x(t)]$ as
\begin{equation}
f(x_t) = f(\mu_t) + \frac{\partial f}{\partial x}\bigg|_{x_t=\mu_t} (x_t - \mu_t) + \mathcal{O}(\| x_t - \mu_t\|^2), \label{eq:linearization_method}
\end{equation}
and let $\mathcal{R}\subseteq \R^n$ denote the region of validity of the above linear approximation. 
When the state distribution exhibits significant deviations from its assumed mean $\mu_t$, the first-order approximation in~\eqref{eq:linearization_method} becomes unreliable, as a non-negligible portion of the probability mass may lie outside $\mathcal{R}$. 
Consequently, these approaches do not handle nonlinear systems with multimodal uncertainty perfectly. 

Recent advances in diffusion-based stochastic control~\cite{berner2022optimal, domingoEnrich2024socm} demonstrate strong empirical performance, but they often lack explicit feedback synthesis and rigorous enforcement of terminal distribution constraints, which are critical in many applications. As a result, a computationally viable framework for stochastic, nonlinear, distribution-aware control with explicit terminal guarantees remains missing, in particular, in the context of trajectory optimization. 

\textit{Contributions:}
We introduce a \emph{multiple distribution-to-distribution linearization} framework that decomposes a nonlinear density steering problem into Gaussian-to-Gaussian OCS subproblems. 
We then show that each subproblem can be locally linearized and efficiently solved via semi-definite programming.
We analyze our approach both theoretically and empirically in an Earth-to-Mars orbit transfer problem, demonstrating improved performance over linearization-based baselines.

\section{Preliminaries} \label{sec:preliminaries}
This work builds upon two research areas in stochastic optimal control: covariance steering and the Schr\"odinger bridge problem (SBP).
In this section, we provide a brief overview of these two research topics and summarize the key results relevant to our proposed approach.

\subsection{Optimal Covariance Steering} \label{sec:OCS}

The optimal covariance steering problem seeks a control policy that minimizes control effort while driving a stochastic time-varying dynamical system from a prescribed initial mean and covariance to a desired terminal mean and covariance, subject to the system dynamics and control constraints.

For continuous-time \emph{linear} stochastic systems with Gaussian initial and terminal distributions, the OCS problem can be formulated as the following optimization problem:
\begin{subequations} \label{eq:OCS}
\begin{align} 
\min_{u \in \mathcal{U}} \quad  &\mathbb{E} \left[ \int_0^T \|u_t \|^2 \d t \right],  \\
 &   \d x_t = A_t x_t \d t + B_t u_t \d t + D_t \d w_t, \label{eq:OCSsystem}\\
& x_0 \sim \mathcal{N}(\mu_0, \Sigma_0), \quad x_T \sim \mathcal{N}(\mu_T, \Sigma_T).
\end{align} \label{eq:full_problem}
\end{subequations}
In the absence of chance constraints on the control and the state, problem \eqref{eq:OCS} admits an optimal solution in the form of the affine state-feedback law \cite{chen2015optimal, chen2015optimal2}
\begin{equation}
    \label{eq:affine_ctrl}
    u_t = K_t (x_t - \mu_t) + v_t, 
\end{equation}
where $K_t \in \mathbb{R}^{m \times n}$ governs the covariance dynamics, while $v_t \in \mathbb{R}^{m}$ is a feedforward term that steers the mean trajectory.
Due to linearity, the mean- and covariance-steering problems are decoupled. 
The feedforward term $v_t$ can be computed in closed form, or expressed as a quadratic program \cite{rapakoulias2023discrete}, and the feedback gain $K_t$  can be computed via the semidefinite program \cite{chen2015optimal2, rapakoulias2023discrete, balci2022exact} by solving
\begin{subequations} \label{eq:OCS_SDP}
\begin{align} 
\min_{\Sigma_t,U_t,Y_t} ~~ & \int_0^T \mathrm{tr}(R_t Y_t) \,\d t, \\ 
 &
\begin{bmatrix} \Sigma_t & U_t^\top \\
U_t & Y_t \end{bmatrix} \succeq 0, \qquad \forall t\in[0,T].
\label{eq:convex_relaxation} \\
&\dot{\Sigma}_t = A_t \Sigma_t + \Sigma_t A_t^\top + B_t U_t + U_t^\top B_t^\top + D_t D_t^\top , \label{eq:cov_dynamics} \\
&\Sigma_0 = \Sigma_t \vert_{t=0}, \quad  \Sigma_{T} = \Sigma_t \vert_{t=T}, \label{eq:prob_obj} 
\end{align}
\end{subequations}
where $K_t = U_t \Sigma_t^{-1}$, and $Y_t$ is a slack variable that equals $K_t \Sigma_t K_t\t$ at optimality.

\subsection{Nonlinear Covariance Steering via Linearization}

When the system dynamics are nonlinear, the mean and covariance steering problems become coupled, rendering the problem significantly more challenging.
Formally, the nonlinear version of the OCS problem is given by
\begin{subequations} \label{eq:N_OCS}
\begin{align} 
\min_{u \in \mathcal{U}} \quad  & \mathbb{E} \left[ \int_0^T \|u_t \|^2 \d t \right],  \\
&   \d x_t = f_t(x_t) \d t + B_t u_t \d t + D_t \d w_t, \label{eq:N_OCS1}\\
& x_0 \sim \mathcal{N}(\mu_0, \Sigma_0), \quad x_T \sim \mathcal{N}(\mu_T, \Sigma_T),
\end{align}
\end{subequations}
where $f_t$ is a state-dependent nonlinear drift function. 
We note that although we keep $ B_t$ and $ D_t$ state-independent in our analysis, a more general analysis in which they are also state-dependent can be easily carried out \cite{ridderhof2019nonlinear}.
While analytical solutions to the nonlinear OCS problem remain open, several approaches have been developed to compute approximate solutions.

For example, \cite{ridderhof2019nonlinear} proposed the iterative Covariance Steering (iCS) approach, where the nonlinear problem is iteratively solved as a sequence of approximate linear covariance steering problems, leveraging the principle of successive convexification \cite{mao2016successive}.
The iCS approach iterates between (i) propagating the nonlinear mean dynamics~\eqref{eq:N_OCS1} under the current control law and (ii) linearizing the system about the resulting mean trajectory.
The algorithm is initialized with an initial guess of the control law $\bar{u}^{(0)}_t$.
At each iteration $k$, the control law $\bar{u}^{(k)}_t$ is used to propagate the nonlinear dynamics and obtain the mean trajectory $\bar{x}^{(k)}_t$.
The nonlinear dynamics are then linearized around $(\bar{x}^{(k)}_t, \bar{u}^{(k)}_t)$.
Introducing the deviation variables $\tilde{x}_t = x_t - \bar{x}^{(k)}_t$ and $\tilde{u}_t = u_t - \bar{u}^{(k)}_t$ yields a local continuous-time linear stochastic system with additive noise of the form
\begin{subequations}
\begin{align}
\d \tilde{x}_t &\approx (A_t^{(k)} \tilde{x}_t + B_t \tilde{u}_t)\, \mathrm{d}t + D_t, \mathrm{d} w_t, 
\quad k = 0,1,\ldots, \label{eq:lin_approx_dynamics} \\
A_t^{(k)} &= \frac{\partial f_t}{\partial x} \bigg|_{x=\bar{x}^{(k)}_t}, \label{A_jacobian}
\end{align}
\end{subequations}
where $A_t^{(k)}$ denotes the Jacobian evaluated along the reference trajectory $(\bar{x}^{(k)}_t, \bar{u}^{(k)}_t)$.
The resulting linearized problem is solved via the SDP in~\eqref{eq:OCS_SDP}, yielding an updated control $\bar{u}^{(k+1)}_t$.
%
This procedure is repeated until convergence. 
Observe that the resulting problem at each iteration is formulated as a convex optimization problem~\eqref{eq:OCS_SDP}, which is solved to obtain the optimal feedforward $v_t$ and feedback gain $K_t$, subject to the linearized dynamics.

\subsection{Gaussian Mixture Schr\"odinger Bridge}
Consider now a generalization of the linear OCS in~\eqref{eq:OCS} where, instead of Gaussian marginals, the initial and final boundary distributions are GMMs with $N_0$ and $N_1$ components, respectively.
Then, the optimization problem is cast as follows
\begin{subequations} \label{eq:GMM_SB}
\begin{align}
\min_{u \in \mathcal{U}} \ \ &J_{\mathrm{GMM}} \triangleq \mathbb{E} \left[ \int_0^{T} \|u_t(x_t) \|^2 \, \d t \right], \\
& \d x_t = A_t x_t  \, \d t + B_t u_t \, \d t + D_t \, \d w_t, \label{eq:sys} \\ 
&\hspace{-0.15in}x_0 \!\sim\! \sum_{i=1}^{N_0} \alpha_0^{i} \mathcal{N}(\mu_0^{i}, \Sigma_0^{i}), \quad x_T \!\sim\! \sum_{j=1}^{N_1} \alpha_T^{j} \mathcal{N}(\mu_T^{j}, \Sigma_T^{j}). \label{eq:GMM_constraints}
\end{align} 
\end{subequations}
An approximation to the solution of \eqref{eq:GMM_SB} can be achieved following the method presented in \cite{rapakoulias2024go}. 
Let $u_{t|ij}$ be the conditional policy that solves the $(i,j)$-OCS problem, i.e., an affine control law of the form \eqref{eq:affine_ctrl}, that steers the system \eqref{eq:sys} from the $i$-th component of the initial Gaussian mixture to the $j$-th component of the terminal mixture, and let $\rho_{t|ij}$ be the resulting Gaussian density of the $(i,j)$-OCS problem. 
A feasible set of control laws for \eqref{eq:GMM_SB} is given by
\begin{equation}
u_t(x) =   \sum_{i,j} u_{t|ij}(x) \frac{\rho_{t|ij}(x) \lambda_{ij}}{\sum_{i',j'} \rho_{t|i'j'}(x) \lambda_{i'j'}}, \label{eq:mixture_policy}
\end{equation}
where $\lambda_{ij}$ is a set of positive weights, obtained through the linear program in \eqref{eq:OT} below (see also \cite[Theorem 1]{rapakoulias2024go}), while the corresponding density of the state is 
\begin{equation} \label{mixture_rho}
    \rho_t = \sum_{ij} \lambda_{ij} \rho_{t|ij}. 
\end{equation}
The calculation of the component-level transport plan $\lambda_{ij}$ is computed by solving
\begin{subequations}
\begin{align}
\min_{\lambda_{ij} \geq 0} \quad &J_{\mathrm{OT}} \triangleq \sum_{i, j} \lambda_{ij} \ C_{ij}, \\
 &\sum_j \lambda_{ij} = \alpha_0^{i},  \quad \forall i = 1,\ldots,N_0, \\ \quad &\sum_i \lambda_{ij} = \alpha_T^{j}, \quad \forall j = 1,\ldots, N_1,
\end{align} \label{eq:OT}
\end{subequations}
where $C_{ij}$ is the cost of the $(i,j)$-OCS subproblem \cite[Theorem 2]{rapakoulias2024go}.
Notice that the proposed weighting scheme prioritizes the conditional policies whose mean is closer to the value of $x_t$.
This formulation incorporates the idea that each $(i,j)$ reference policy is associated with a prior probability $\lambda_{ij}$ induced by the optimal transport coupling, while the state $x_t$ at time $t$ has a known likelihood of being generated from the Gaussian component $\rho_{t|ij}$. 
Similar concepts have also been explored in discrete-time density control problems using GMMs \cite{balci2023density, kumagai2024chance, rapakoulias2026schrodinger}.

\section{Nonlinear Density Steering}

In this section, we extend the techniques presented in Section~\ref{sec:preliminaries} to approximate the nonlinear density steering problem. 
While, in general, a final distribution constraint increases the complexity of solving an SOC problem, in the proposed formulation, we exploit this additional structure stemming from the GMM boundary conditions to construct a set of reference linearizations that synthesize a non-linear control law.

\subsection{Multiple Distribution-to-Distribution Linearization}

Let $\rho_0$ and $\rho_T$ be the initial and final boundary distributions on $\mathbb{R}^n$. 
Since Gaussian mixture models are dense in the space of probability measures, any boundary distributions can be approximated arbitrarily well by GMMs~\cite{mazya2007approximate}.
To this end, we restrict our attention to the case where $\rho_0$ and $\rho_T$ are given GMMs, and consider the problem
\begin{figure}[t]
\centering
\includegraphics[width=0.6\columnwidth]{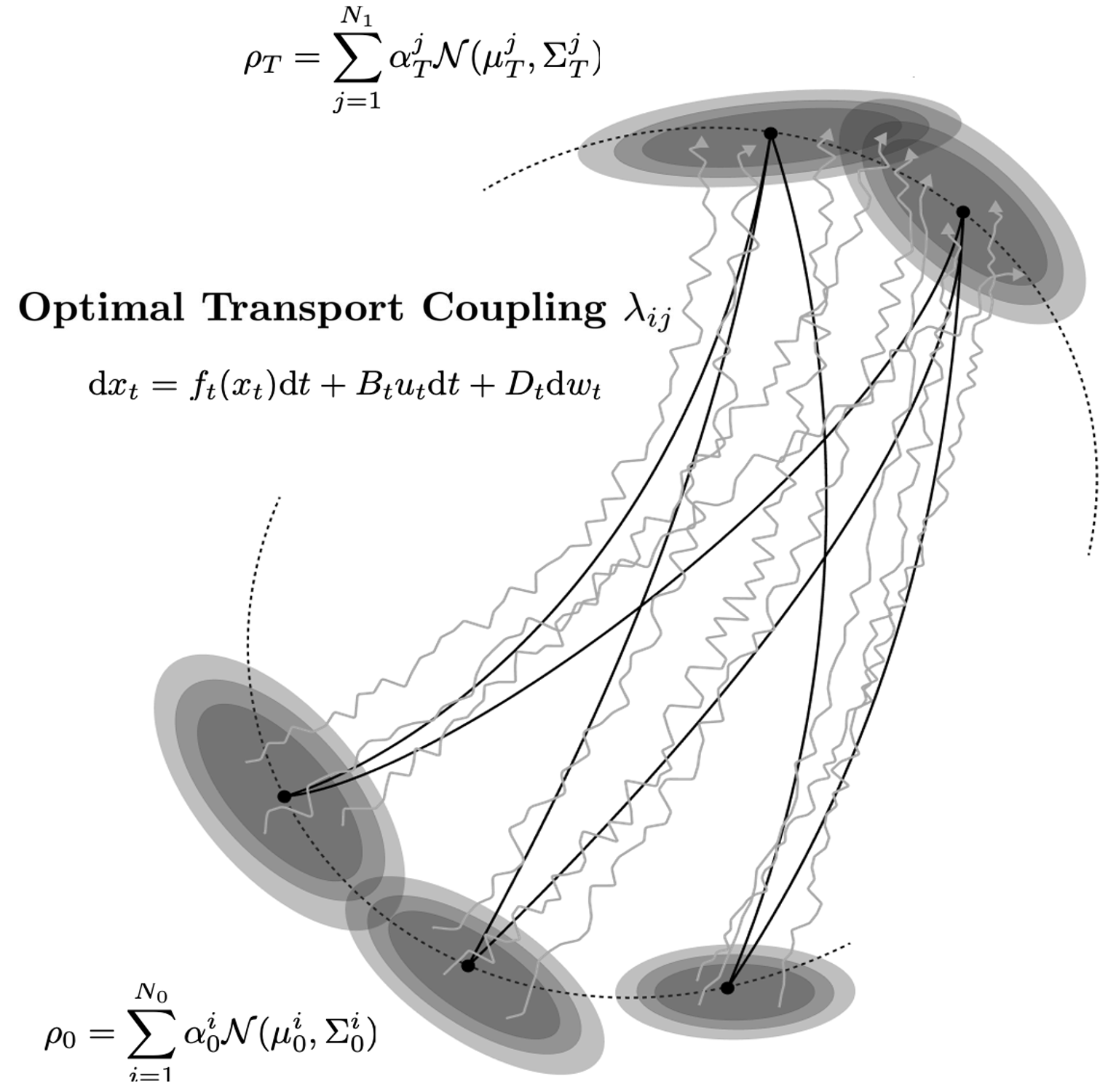}
\caption{The problem of steering the initial distribution $\rho_0$, modeled as a Gaussian Mixture, of the state of a stochastic continuous-time non-linear system to a desired terminal distribution $\rho_T$ (GMM), at a given stage $t=T$. 
Each ellipse in a different intensity of gray represents a single Gaussian distribution, where a darker color highlights the region of more probability mass.}
\label{fig:description}
\end{figure}
\begin{subequations} \label{eq:N_OCS_GMM}
    \begin{align}
    \min_{u \in \mathcal{U}} \quad & \mathbb{E}_{x_t} \left[ \int_0^{T} \| u_t (x_t)\|^2 \d t\right], \\
  & \d x_t = f_t(x_t) \d t + B_t u_t \d t + D_t \d w_t,   \label{eq:nonlinear_eq} \\
    &\hspace{-0.25in}x_0 \!\sim\! \sum_{i=1}^{N_0} \alpha_0^{i} \mathcal{N}(\mu_0^{i}, \Sigma_0^{i}), \quad x_T \!\sim\! \sum_{j=1}^{N_1} \alpha_T^{j} \mathcal{N}(\mu_T^{j}, \Sigma_T^{j}).
\end{align}
\end{subequations}
Following \cite{rapakoulias2024go, rapakoulias2025steering}, we construct a policy similar to \eqref{eq:mixture_policy} by exploiting its mixture structure to linearize \eqref{eq:nonlinear_eq} around different reference trajectories for each conditional policy.
Fig.~\ref{fig:description} illustrates the high-level idea of this approach. 

We now formally define the building blocks of the multiple linearization approach.
For each component pair $(i,j) \in \{1,\ldots,N_0\} \times \{1,\ldots,N_1\}$, the
\emph{reference mean trajectory} $\mu_{t|ij}$ is defined as the solution of the
\textit{deterministic} open-loop optimal control problem
\begin{subequations} \label{eq:mean_ref} 
\begin{align} \min_{\bar{u}_{t|ij}} \quad & J_{ij} \;=\; \int_0^T \| \bar{u}_{t|ij} \|^2 \, \d t, \\ 
&\dot{\mu}_{t|ij} = f_t(\mu_{t|ij}) + B_t \bar{u}_{t|ij}, \\ &\mu_{0|ij} = \mu^i_0, \quad \mu_{T|ij} = \mu^j_T, 
\end{align} 
\end{subequations}
where $\bar{u}_{t|ij}$ is the open-loop control obtained by solving the deterministic nonlinear program~\eqref{eq:mean_ref} with boundary conditions fixed at the $i$-th initial and $j$-th
terminal GMM component means. 

With the optimal mean trajectory $\mu_{t|ij}$ and mean control $\bar u_{t|ij}$ known, we evaluate the Jacobian matrix of the nonlinear drift term in \eqref{eq:nonlinear_eq} around the reference mean trajectories using Equations \eqref{A_jacobian}.
The above procedure yields a set of matrices $A_{t|ij}, B_{t}$ and $D_{t}$, where $t \in [0, T]$.
 
Given the linearized system~\eqref{eq:lin_approx_dynamics}, the associated $(i,j)$-OCS problem is then solved as described in Section~\ref{sec:OCS}, producing an optimal affine feedback law of the form
\begin{equation}
    u_{t|ij}(x) = K_{t|ij}\bigl(x - \mu_{t|ij}\bigr) + \bar{u}_{t|ij},
    \label{eq:control_ij}
\end{equation}
where $K_{t|ij} \in \mathbb{R}^{m \times n}$ is the optimal feedback gain controlling the covariance dynamics, and $\bar{u}_{t|ij}$ is the feedforward term steering the mean along the reference trajectory. 
The density of the  $(i,j)$-th Gaussian bridge is given by 
\begin{equation}
    \rho_{t|ij} = \mathcal{N}\!\bigl(\mu_{t|ij},\Sigma_{t|ij}\bigr),
\end{equation}
and satisfies $\rho_{0|ij} = \mathcal{N}(\mu^i_0,\Sigma^i_0)$ and
$\rho_{T|ij} = \mathcal{N}(\mu^j_T,\Sigma^j_T)$. 

Repeating this procedure for each $(i,j)$ of the $N_0 \times N_1$ subproblems allows us to define the necessary building blocks to compute the control law defined by \eqref{eq:mixture_policy} as a mixture of multiple linearizations, and hence we refer to this method as a \textit{multiple linearization} approach. 
In particular, the conditional densities serve as the weighting kernels, while the final nonlinear controller is given by
\begin{equation} 
    u_t(x)
    = \sum_{i,j}
      \bigl(K_{t|ij}\bigl(x - \mu_{t|ij}\bigr) + \bar{u}_{t|ij}\bigr)
      \frac{\rho_{t|ij}(x)\,\lambda_{ij}}
           {\displaystyle\sum_{i',j'}\rho_{t|i'j'}(x)\,\lambda_{i'j'}},
    \label{eq:mixture_control_explicit}
\end{equation}
where the weights $\lambda_{ij}$ are obtained by solving the linear
program~\eqref{eq:OT}. 
This formulation ensures that, at any point $x_t$ in the state-space, the overall control is a linear combination of the local affine policies~\eqref{eq:control_ij}, with each policy weighted by the posterior probability that $x_t$ was generated by the $(i,j)$-th Gaussian component $\rho_{t|ij}$
at time $t$.

We note that, in practice, this calculation is performed in a temporal discretization of the time horizon $[0, T]$. 
Since a sufficiently small time step yields negligible discretization error, we perform the analysis in continuous time and leave an inherently discrete-time analysis for future work.

\subsection{Error Analysis and Density Characterization} \label{sec:analysis}



We now evaluate the performance of the proposed multiple linearization approach in comparison to the single-linearization methods in \cite{ridderhof2019nonlinear, ridderhof2020chance, benedikter_convex_2022}.
Before formally assessing this claim, we first study an abstract problem from a function approximation perspective to understand how multiple linearizations can improve the accuracy of approximating a nonlinear function.

Consider a Gaussian mixture $\rho = \sum_i \alpha_i \N(\mu_i, \Sigma_i)$ with $N$ components and mean $\mu=\sum_i \alpha_i\mu_i$.
We consider two approximations of a differentiable function $f: \R^n \to \R^n$. 
\begin{align}
    f_\mathrm{SL}(x) &= f(\mu)+ \nabla f(\mu)(x-\mu), \label{eq:SL_ref} \\
    f_{\mathrm{ML}}(x) &= \sum_i  \frac{\alpha_i \N(x; \mu_i, \Sigma_i)}{\rho(x)} f_{i}(x), \label{eq:ML_ref}
\end{align}
where 
\begin{equation}
f_i(x) \triangleq f(\mu_i) + \nabla f(\mu_i)(x - \mu_i) \label{eq:FI}    
\end{equation} 
denotes the linearization of $f$ around the mean $\mu_i$ of the $i$-th mixture component.

In words, $f_\mathrm{SL}$ represents the single linearization (SL) approximation around the global mean $\mu$,
while $f_\mathrm{ML}$ is a weighted average of multiple linearizations (ML) around each component's mean $\mu_i$. 

We now state a result that characterizes the upper bounds on the expected error for these two approximations.

\begin{theorem}
    \label{theo2}
    Suppose the linearization error of $f$ is uniformly quadratically bounded. 
    Then, the upper bound on the expected error of the multiple linearization, $\E \big( \|f - f_{\mathrm{ML}}\|_2\big) \le \epsilon_{\mathrm{ML}}$, is no worse than the corresponding upper bound of the single linearization, $\E \big( \|f - f_{\mathrm{SL}}\|_2\big) \le \epsilon_{\mathrm{SL}}$, that is, $\epsilon_{\mathrm{ML}} \le \epsilon_{\mathrm{SL}}$.
\end{theorem}

\begin{proof}
Since, by assumption, the linearization error of $f$ is uniformly quadratically bounded, for all $x_0,x \in \R^n$, we have
    \begin{equation}
    \label{eqn:quadratic-bound}
        \|f(x) - \bar{f}_{x_0}(x)\|_2 \le C \|x-x_0\|_2^2,
    \end{equation}
where 
$\bar f_{x_0}(x) \triangleq f(x_0) + \nabla f (x_0) (x - x_0)$ is the first-order approximation around $x_0$, and $C$ is a constant. 
For component-wise $L$-smooth functions, the quadratic error bound~\eqref{eqn:quadratic-bound} follows from~\cite[Eq.~(8.3)]{nocedal2006numerical}. 

For the single linearization approximation, the total expected approximation error is bounded as
\begin{subequations}
\begin{align}
 \E \big( \|f(x) - f_{\mathrm{SL}}(x)\|_2 \big)
&= \! \int_{\R^n} \!\!\| f(x) - f_{\mathrm{SL}}(x)\|_2\rho(x) \, \d x  \\
&\leq C \, \int_{\R^n} \| x - \mu\|^2_2 \; \rho(x) \, \d x 
=  C\, \tr ( \Sigma ) \triangleq \epsilon_{\mathrm{SL}}, \label{eq:sigma}
\end{align}
\end{subequations}
where $\Sigma$ is the covariance of $\rho$.
Similarly, the multiple-linearization approximation yields the following upper bound on the expected error
\begin{subequations}
\label{eqn:ML-bound}
\begin{align}
\int_{\R^n} \| f(x) - f_{\mathrm{ML}}(x) \|_2 \; \rho(x)\, \d x &=  \int_{\R^n} \ \biggl\| \sum_{i=1}^N \frac{\alpha_i \mathcal{N}(x ; \mu_{i}, \Sigma_{i})}{\rho(x)}(f(x) - f_i(x)) \biggl\|_2 \; \rho(x)\, \d x \\
& \leq \sum_{i=1}^N \alpha_i \int_{\R^n} \|f(x) - f_{i}(x) \|_2 \; \mathcal{N}(x ; \mu_{i}, \Sigma_{i}) \,\d x \label{eq:jensen1} 
\\
&\leq \sum_{i=1}^N \alpha_i C \, \int_{\R^n} \| x - \mu_i\|_2^2 \; \mathcal{N}(x ; \mu_{i}, \Sigma_{i}) \, \d x \label{eq:bound_ML_1} \\
&= C \sum_{i=1}^N \alpha_i \tr ( \Sigma_{i} ) \triangleq \epsilon_{\mathrm{ML}}, \label{eq:bound_ML}
\end{align}
\end{subequations}
where \eqref{eq:jensen1} is from the discrete Jensen's inequality \citep[Theorem~7.35]{wheeden1977measure} and \eqref{eq:bound_ML_1} follows from~\eqref{eqn:quadratic-bound}.
Using the variance decomposition, we obtain 
\begin{equation} \label{eq:var_decomp}
\Sigma = \sum_{i} \alpha_{i} \Sigma_{i} + \sum_{i}\alpha_{i} (\mu_{i} - \mu) (\mu_{i} - \mu)\t.
\end{equation}
Taking the trace on both sides of \eqref{eq:var_decomp} leads to
\begin{equation}
\tr  ( \Sigma )  =  \sum_{i} \alpha_i \tr ( \Sigma_{i} ) + \sum_{i} \alpha_i \| \mu_{i} - \mu\|_2^2,  \label{eq:trace_exp}
\end{equation}
which implies that $\tr (\Sigma)  \geq \sum_{i} \alpha_i \tr (\Sigma_{i}$).
The result then follows immediately from the definitions of $\epsilon_{\mathrm{SL}}$ and $\epsilon_{\mathrm{ML}}$ in \eqref{eq:sigma} and \eqref{eq:bound_ML}, respectively.
%
\end{proof}

Unfortunately, we note that an improved upper bound on the linearization error between the ML and SL approximations does not, in general, imply better approximation accuracy, since the derived bounds may be loose.
A natural question, therefore, is under what conditions on $f$ can we obtain a stronger result.
It turns out that the convexity or concavity of the components of the vector field $f$ is sufficient to prove strictly better approximation accuracy between the multi and single linearization approaches.
We establish this in Theorem \ref{theo3}.
In the following, to facilitate a component-wise analysis, we will compare linearization accuracy using the expected $1$-norm error, denoted by $\| \cdot \|$.

\begin{lemma}
\label{lem:optimal_linearization}
Consider the $i$-th component of the GMM $\rho$ and assume each component $f^{(k)}: \mathbb{R}^n \to \mathbb{R}$, 
$k \in \{1,\ldots,n\}$ of the vector field $f$, is either globally convex or concave. 
Then, denoting by $f_a$ the linearization of $f$ around $a \in \mathbb{R}^n$, it holds that
\begin{equation} \label{lemma_eq}
\begin{split}
\int_{\R^n} \|f(x)&-f_i(x)\| \mathcal{N}(\mu_i, \Sigma_i)\, \d x \leq \int_{\R^n} \|f(x) - f_a(x)\| \mathcal{N}(\mu_i, \Sigma_i) \, \d x.
\end{split}
\end{equation}
\end{lemma}
for all values of $a\in \R^n$, where $f_i$ is given by Equation \eqref{eq:FI}.
\begin{proof}

Expanding the 1-norm in the LHS of \eqref{lemma_eq}, we obtain
\begin{equation} \label{expanded_1norm}
\int_{\R^n} \|f(x) - f_{i}(x)\| \mathcal{N}(\mu_i, \Sigma_i) \, \d x  = \sum_{k=1}^n \int_{\R^n} |f^{(k)}(x) - f^{(k)}_{i}(x)| \mathcal{N}(\mu_i, \Sigma_i) \, \d x 
\end{equation}
We will study each term in the sum \eqref{expanded_1norm} individually. 
To this end, consider a single component $f^{(k)}$, and define the linearization around a point $\alpha \in \R^n$ as $f^{(k)}_a(x) = f^{(k)}(a) + \nabla f^{(k)}(a)\t(x - a)$. 
Furthermore, define the component linearization error around the distribution $\N(\mu_i, \Sigma_i)$ as 
\begin{equation} \label{component_error}
g^{(k)}(a) = \mathbb{E}_{x \sim \mathcal{N}(\mu_i, \Sigma_i)}\left[|f^{(k)}(x) - f^{(k)}_a(x)|\right].
\end{equation}

Assume, for now, that $f^{(k)}$ is convex. 
From properties of convex functions \cite{boyd2004convex}, it holds that all first-order approximations of $f^{(k)}$ are global underestimators, that is $f^{(k)}(x) \geq f^{(k)}_a(x)$ for all $x, a \in \R^n$. 
Given the last inequality, we will show that $g^{(k)}(\mu_i) \leq g^{(k)}(a) $  for any $a \in \R^n$, that is, the linearization error in \eqref{component_error} is minimized when the function is linearized around the underlying distribution's mean. 
More specifically, the following statements hold
\begin{align*}
   g^{(k)}(\mu_i) \leq g^{(k)}(a) 
 \iff & \int_{\R^n} |f^{(k)}(x) -f^{(k)}_{i} (x)| \mathcal{N}(\mu_i, \Sigma_i)\, \d x \leq \int_{\R^n} |f^{(k)}(x) - f_a(x) | \mathcal{N}(\mu_i, \Sigma_i) \, \d x \\
 \iff & \int_{\R^n} \big( f^{(k)}(x) -f^{(k)}_{i} (x) \big) \mathcal{N}(\mu_i, \Sigma_i)\, \d x \leq \int_{\R^n} \big( f^{(k)}(x) - f_a(x) \big) \mathcal{N}(\mu_i, \Sigma_i) \, \d x \\
 \iff & \int_{\R^n} \big( f^{(k)}_{i}(x) -f^{(k)}_{a} (x) \big) \mathcal{N}(\mu_i, \Sigma_i)\, \d x \geq 0 \\
\iff & \mathbb{E}_{x \sim \mathcal{N}(\mu_i, \Sigma_i)}\left[ f_i^{(k)}(x) \right] - \mathbb{E}_{x \sim \mathcal{N}(\mu_i, \Sigma_i)}\left[ f^{(k)}_a(x)\right] \geq 0 \\
 \iff & f^{(k)} (\mu_i) + \nabla f^{(k)}(\mu_i) (x - \mu_i)\big|_{x= \mu_i} - f^{(k)} (a) - \nabla f^{(k)}(a) (x - a)\big|_{x= \mu_i} \geq 0 \\
 \iff & f^{(k)}(\mu_i) - f^{(k)}(a) - \nabla f^{(k)} (a) ( \mu_i - a) \geq 0 \\
 \iff &  f^{(k)}(\mu_i) \geq f^{(k)}_a (\mu_i).
\end{align*}
%
A similar analysis holds when the component $f^{(k)}$ is concave. 
Returning to \eqref{expanded_1norm}, we have 
\begin{align}
\int_{\R^n} \|f(x) - f_{i}(x)\| \mathcal{N}(\mu_i, \Sigma_i) \, \d x 
= \sum_{k=1}^n g^{(k)}(\mu_i) 
 \leq \sum_{k=1}^n g^{(k)}(a) 
= \int_{\R^n} \|f(x) - f_a(x)\| \mathcal{N}(\mu_i, \Sigma_i)\, \d x, 
\end{align}
which yields the desired result.
\end{proof}

We now prove that under the assumptions of Lemma \eqref{lem:optimal_linearization}, ML is a better approximation of $f$ compared to SL: 

\begin{theorem} \label{theo3}
Let $f_{\mathrm{SL}}, f_{\mathrm{ML}}$ be defined as in \eqref{eq:SL_ref}, \eqref{eq:ML_ref}, respectively, and $f: \R^n \to \mathbb{R}^n$ satisfy the conditions of Lemma \ref{lem:optimal_linearization}. Then, the multiple linearization approximation $f_{\mathrm{ML}}$ has better linearization error compared to $f_{\mathrm{SL}}$, that is,  
\begin{equation}
     \E \big( \|f - f_{\mathrm{ML}}\| \big) \leq \E \big( \|f - f_{\mathrm{SL}}\| \big). 
\end{equation}
\end{theorem}

\begin{proof}
We can decompose the expected error as
\begin{subequations}
\begin{align}
\E \big( \|f - f_{\mathrm{ML}}\| \big) &=  \int_{\R^n} \| f(x) - f_{\mathrm{ML}}(x) \| \rho(x) \, \d x, \\ \label{eq:jensen_step}
&\leq \sum_{i=1}^N \alpha_{i} \int_{\R^n} \|f(x) - f_{i}(x) \| \mathcal{N}(\mu_{i}, \Sigma_{i}) \, \d x,  \\ \label{lemma_step} 
&\leq  \sum_{i=1}^N \alpha_{i} \int_{\R^n} \|f(x) - f_{\mathrm{SL}}(x) \| \mathcal{N}(\mu_{i}, \Sigma_{i}) \d x  \\
&= \E \big( \|f - f_{\mathrm{SL}}\| \big)
\end{align}
\end{subequations}
where \eqref{eq:jensen_step} follows from the discrete Jensen's inequality \citep[Theorem~7.35]{wheeden1977measure}, and \eqref{lemma_step} follows from Lemma \ref{lem:optimal_linearization}. 
This concludes the proof of Theorem \ref{theo3}.

\end{proof}

We now return to Problem \eqref{eq:N_OCS_GMM}.
We verify that, under the multiple-linearization policy,
the mixture density $\rho_t$ satisfies an FPK equation \ref{eq:FPK} whose drift approximates the
true nonlinear drift $f_t$.
Specifically, consider the FPK equation corresponding to the
$(i,j)$-linearization, with drift
$\bar{f}_{t|ij}(x)=f_t(\mu_{t|ij})+A_{t|ij}(x-\mu_{t|ij})$, given by
\begin{equation}
    \frac{\partial \rho_{t|ij}}{\partial t}
    + \nabla\cdot\!\Bigl(\rho_{t|ij}
      \bigl(\bar{f}_{t|ij} + B_t u_{t|ij}\bigr)\Bigr)
    - \frac{1}{2}\operatorname{tr} \bigl(D_t D_t\t
      \nabla^{2}\rho_{t|ij}\bigr) = 0,
\end{equation}
where $u_{t|ij}$ solves the $(i,j)$-OCS subproblem.
Multiplying by $\lambda_{ij}$ and summing over all $(i,j)$
pairs yields
\begin{equation} \label{ml_fpk}
    \frac{\partial \rho_{t}}{\partial t} + \nabla \cdot \left( \left( f_{\mathrm{ML}} + B_t u_t \right) \rho_t \right) - \frac{1}{2} \tr (D_t D_t\t \nabla^2 \rho_{t}) = 0, 
\end{equation}
where, 
\begin{equation}
    f_{\mathrm{ML}} = \sum_{ij} \bar f_{t|ij} \frac{\lambda_{ij} \rho_{t|ij}}{\rho_t}, 
\end{equation}
with
$\rho_t = \sum_{ij} \lambda_{ij} \rho_{t|ij}$, and $u_t$ is defined by \eqref{eq:mixture_control_explicit}.
Equation \eqref{ml_fpk} shows that if the system dynamics were evolving according to $f_{\mathrm{ML}}$, the controller \eqref{eq:mixture_control_explicit} would indeed transport the GMM $\rho_0$ to $\rho_T$ \textit{exactly}, serving as a foundation for our motivation. 
Furthermore, since according to Theorems~\ref{theo2}, and~\ref{theo3}, 
$f_\mathrm{ML}$ is a better approximation of the nonlinear function $f$ than a single linearization, we expect the true $\rho_T$ to be closer to the terminal GMM than with a linear controller. 

Finally, we remark that the component-wise convexity requirement in Theorem~\ref {theo3} to guarantee improved approximation accuracy might be too strong for many practical problems, including our case study on spacecraft navigation in the following section.
Empirically, however, we observed that the ML method continues to outperform SL even when the convexity conditions are not satisfied, as demonstrated in our numerical simulations.


\section{Numerical Simulation}
We evaluate the efficacy of the proposed multiple linearization method in the Earth-Mars rendezvous scenario presented in \cite{benedikter_convex_2022}, under nonlinear orbital dynamics and non-Gaussian state uncertainty. 

\subsection{Modeling of the dynamics}
The dynamics are expressed in a Sun-centered inertial frame restricted to the ecliptic plane (2D). 
Assume the spacecraft has approximately constant mass and leaves Earth with zero hyperbolic excess velocity. 
The physical state is defined by position and velocity expressed in Cartesian components, namely $r_t = [p^x_t, p^y_t]\t, v_t = [v^x_t, v^y_t]\t$. 
Assume the spacecraft is modeled as a point mass and the applied control is the acceleration provided by the thrusters $u_t = [u^x_t,u^y_t]\t$. 
In state-space form, the dynamics can be written as:
\begin{align}
\d x_t = 
\begin{bmatrix} v_t \\ - \dfrac{\mu_{\odot}}{\|r_t\|^3} r_t \end{bmatrix} \d t + \begin{bmatrix} 0_2 \\ I_2 \end{bmatrix} u_t \, \d t + \begin{bmatrix} 0_2 \\ g_v I_2 \end{bmatrix} \d w_t ,\label{dynamics_eq}
\end{align}
where $x_t = [p^x_t, p^y_t, v^x_t, v^y_t ]\t$. 
Note that the final time $T$ is fixed and known, and the Sun's gravity and the noise are the only external forces acting on the system. 
All variables and parameters are scaled to reduce disparities in their orders of magnitude, which is essential for the convergence of the algorithm. Scaling is performed relative to the astronomical unit (AU) and days.
The related linearized representation is:
\begin{align*}
&A_t = \frac{\partial f}{\partial x} = 
\begin{bmatrix}
    0_{2} & I_{2} \\
    A_{vr} & 0_{2} \\
\end{bmatrix}, 
\qquad 
B_t =  \begin{bmatrix}
    0_{2} \\ I_{2}
    \end{bmatrix},
\qquad
D_t =  \begin{bmatrix}
    0_{2} \\  g_v I_{2}
    \end{bmatrix}, 
    \end{align*}
where,  
\begin{align*}
A_{vr} = 
\mu_{\odot} \begin{bmatrix}
   \dfrac{3 (p^{x})^2 - \|r\|^2}{\|r\|^5} & \dfrac{3 p^x p^y}{\|r\|^5} \\[10pt]
    \dfrac{3 p^x p^y}{\|r\|^5} & \dfrac{3 ({p^y})^2 - \|r\|^2}{\|r\|^5}     
\end{bmatrix}.
\end{align*}
In practice, we replace the process noise covariance  $D_t D_t \t$, with 
\begin{align}
\hat{\Sigma}_d = \begin{bmatrix}
    0_{2} & 0_{2} \\
    0_{2} & g_v^2 \ I_{2} \\
\end{bmatrix} + \varepsilon I_{4},
\end{align}
where $\varepsilon>0$ models an artificial noise induced by the linearization.

The nonlinear program (NLP) used to compute the open-loop reference is solved using Sequential Least Squares Quadratic Programming (SLSQP). To provide an accurate initial guess, we adopt a Hohmann transfer estimate, motivated by the well-known sensitivity of this class of algorithms to initialization. Finally, the $N_0 \times N_1$ OCS subproblems are solved using MOSEK \cite{aps2020mosek}.

\begin{table}[!ht]
\centering
\caption{Initial values and simulation parameters.$^{\star}$}
\label{tab:simulation_parameters}
\setlength{\tabcolsep}{3pt}
\renewcommand{\arraystretch}{1.15}
\small
\begin{tabular}{l c p{3.8cm}}
\toprule
\textbf{Parameter} & \textbf{Symbol} & \textbf{ Value} \\

\midrule
\multicolumn{3}{c}{\textit{Simulation Parameters}} \\
\midrule
Init. Components   & $N_0$               & $3$ \\
Target Components  & $N_1$               & $2$ \\
Time Nodes         & $N$                 & $101$ \\
Diffusion Scale    & $\varepsilon^2$     & {$10^{-7}\,\mathrm{diag}([0.05, 0.05,0.1, 0.1])$} \\
Noise Intensity & $g_v$ & $10^{-4}\ \mathrm{m/s^{3/2}}$\\
Sun Grav. Param.   & $\mu_{\odot}$       & $2.9591{\times}10^{-4}\;\mathrm{AU^3/day^2}$ \\
\midrule
\multicolumn{3}{c}{\textit{Boundary Conditions -- Case I}} \\
\midrule
Init. Covariances   & $\Sigma_0^i$      & {$\mathrm{diag}(4.5{\times}10^{-9},\ 4.5{\times}10^{-9}, 3.5{\times}10^{-9}, 3.5{\times}10^{-9})$} \\[4pt]
Target Covariances  & $\Sigma_{T}^j$      & {$\mathrm{diag}(4.5{\times}10^{-6}, \ 4.5{\times}10^{-6},3.5{\times}10^{-8}, 3.5{\times}10^{-8})$} \\
Weights & $\alpha_0^i, \alpha_T^j$ & $1/N_0$, $1/N_1$\\
\midrule
\multicolumn{3}{c}{\textit{Boundary Conditions -- Case II}} \\
\midrule
Init. Mean         & $\mu_0$         & \makecell[l]{$(-0.9405,\,-0.3450)\ \mathrm{AU}$\\$(9.7746,\,-28.078)\ \mathrm{km/s}$} \\[4pt]
Target Mean        & $\mu_T$         & \makecell[l]{$(-1.1543,\;\phantom{-}1.1829)\ \mathrm{AU}$\\$(-16.427,\,-14.861)\ \mathrm{km/s}$} \\[4pt]
Init. Covariance   & $\Sigma_{0}$      & {$\mathrm{diag}(10^{-3},\ 10^{-3},10^{-7}, 10^{-7})$} \\
Target Covariance  & $\Sigma_{T}$      & 
$\mathrm{diag} (10^{-3},\ 10^{-3},\ 10^{-7},\ 10^{-7})$ \\
\bottomrule
\multicolumn{3}{l}{\scriptsize $^{\star}$Reference scales and simulation parameters are common to both cases.}\\
\end{tabular}
\end{table}

To properly evaluate the efficacy of the proposed method, Monte Carlo simulations were conducted considering two primary guidance policies: SL and ML. In the following two subsections, we present two different variations of the problem described here.

\subsection{Case I: Multi-Modal Boundary Distributions}

The first case of interest we investigated is the behavior of both algorithms under multimodal distributions. 
This corresponds to a situation in which multiple starting points and multiple targets can be reached with prescribed probability and accuracy. 
Formally, this corresponds to a case with $N_0$ Gaussian components, each representing a possible initial state, with an associated covariance describing the uncertainty around that state. 
\begin{figure}[htbp]
\centering
\includegraphics[width=0.7\columnwidth]{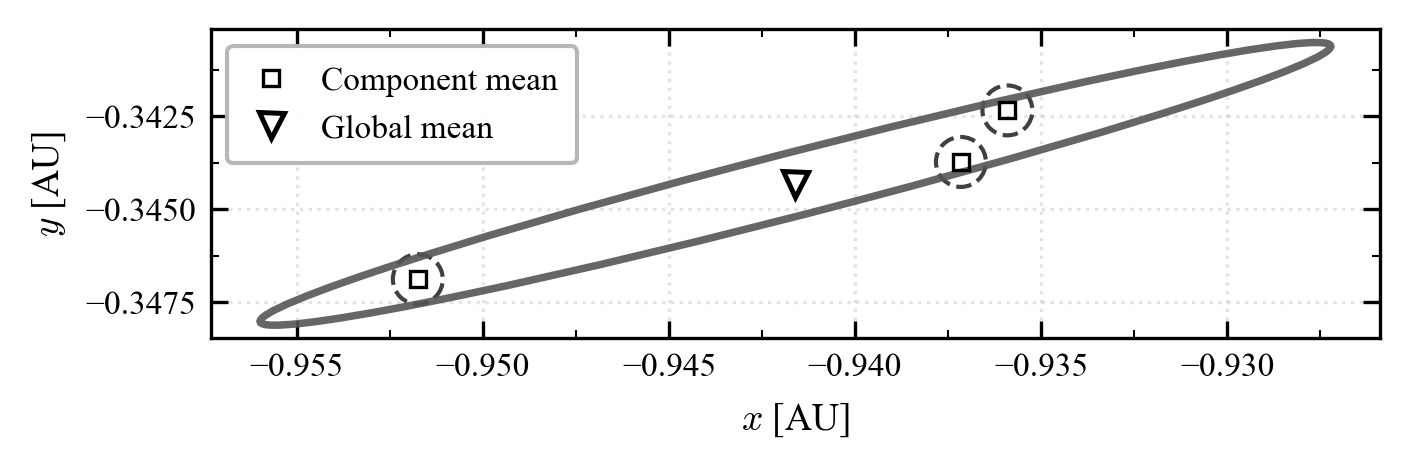}
\caption{Initial multi-modal distribution $\rho_0$. Covariance ellipses associated with the Gaussian components ($i = 1,2,3$) are computed at the two-standard-deviation level and scaled by a factor of $5$ for clarity.}
\label{fig5}
\end{figure}
A similar interpretation can be given for the multiple-arrival condition.
To this end, the initial and final distributions $\rho_0$ and $\rho_T$ are chosen to be composed of three (Figure \ref{fig5}) and two (Figure \ref{fig3}) isolated Gaussian components, respectively, such that their corresponding covariance ellipses are far from intersecting.

\begin{figure}[htbp]
\centering
\includegraphics[width=0.7\columnwidth]{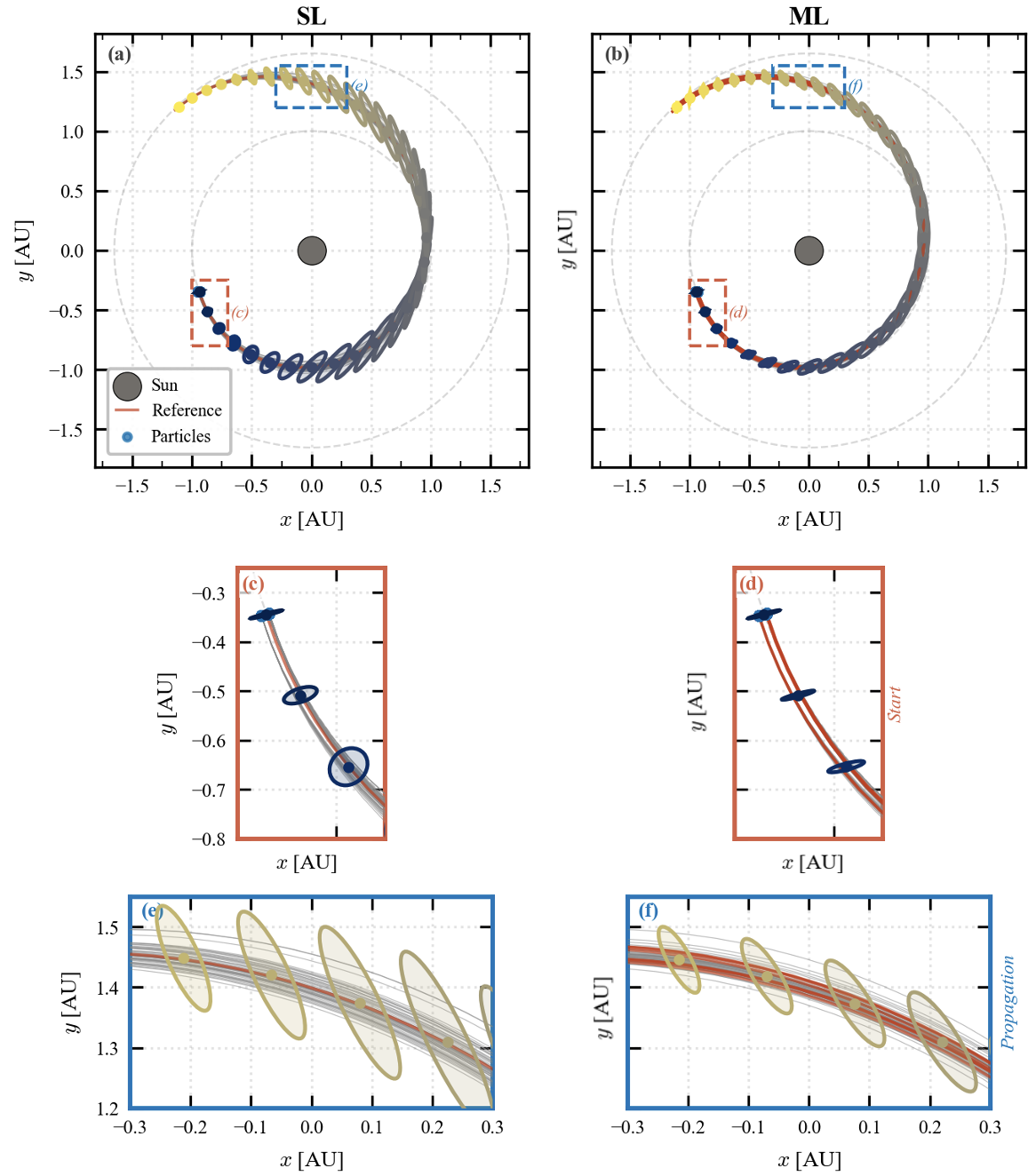}
\caption{Case I: position evolution under SL $(a)$ and ML $(b)$ policies. The top row depicts phase-space trajectories from $t = 0$ (dark blue) to $t = T$ (yellow), showing the evolution of $200$ Monte Carlo samples (gray). Panels $(c)$--$(d)$ illustrate the propagation of the initial samples drawn from $\rho_0$, shown in light blue. The red curves denote the reference trajectories (bridges). $(e)$--$(f)$ provide zoomed-in views highlighting an intermediate time window, illustrating tighter covariance control under ML. The ellipses represent $3\sigma$ confidence regions computed from the empirical distribution.}
\label{fig1}
\end{figure}

The ML approach exploits the superposition of local affine policies to effectively capture the underlying local nonlinearities and the distribution's sparse concentration. 
In contrast, the SL approach solves a single mean-to-mean problem, from the initial global mean to the final global mean, and is therefore unable to capture the separation of the probability mass, whereas ML successfully addresses this feature.
The obtained solutions are presented in Figures~\ref{fig1} 
and~\ref{fig3}.

\begin{figure}[htbp]
\centering
\includegraphics[width=0.50\columnwidth]{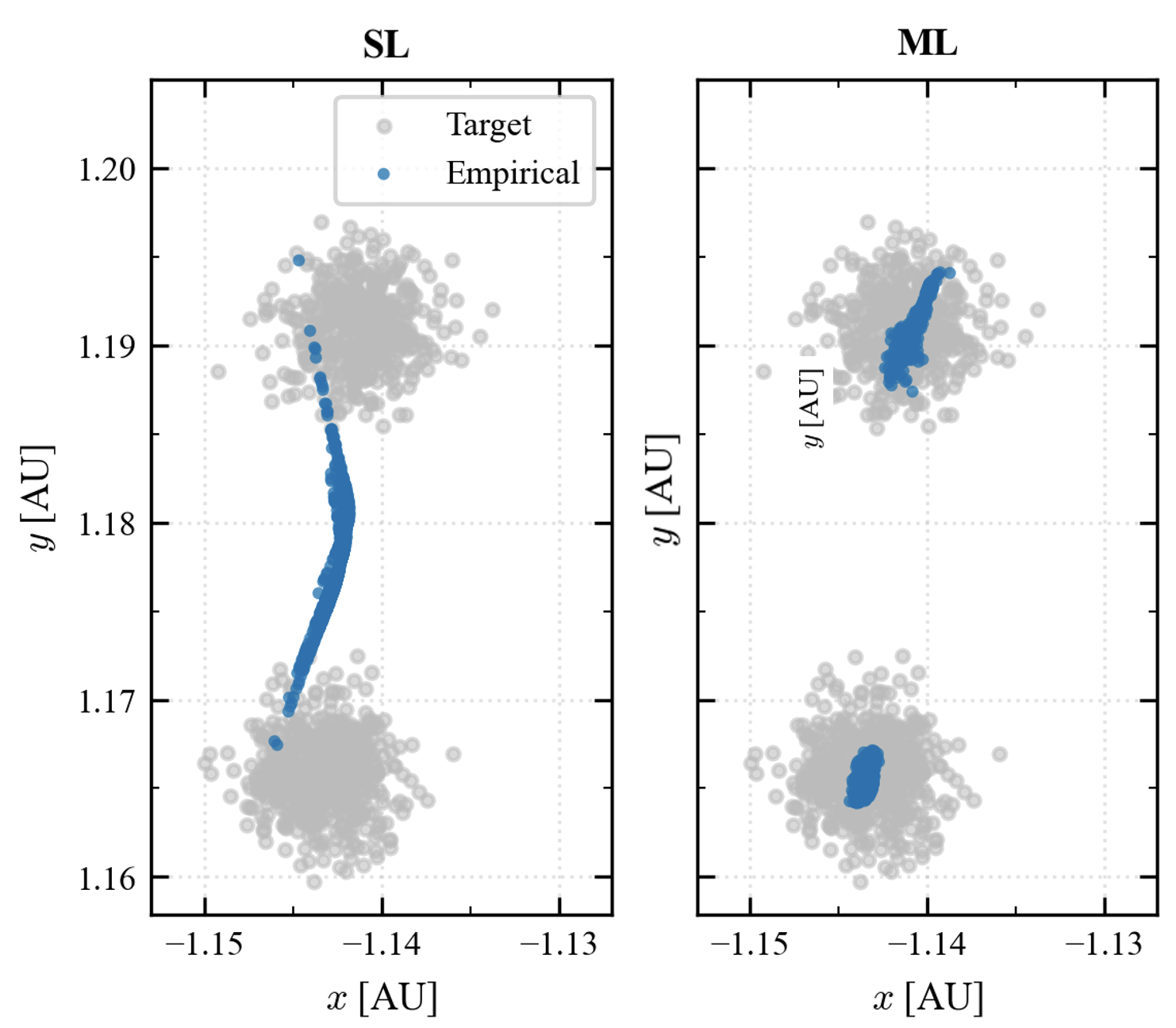}
\caption{Case I: terminal distribution matching.} 
\label{fig3}
\end{figure}

\subsection{Case II: Exploitation of Multiple Linearizations}

Consider a different scenario where the marginal constraints correspond to distributions with large covariances. 
Specifically, we consider the distributions in Figures \ref{fig7} and \ref{fig6}, with nominal values given in Table \ref{tab:simulation_parameters}. 
By artificially fitting a GMM to $\rho_0$ and $\rho_T$, we can construct a more accurate approximation of the nonlinear controller by leveraging the concept of multiple linearizations.  
We evaluate this claim by comparing the resulting control effort of the proposed approach, computed via Monte Carlo simulations, with the control cost of a single linearization.
While both methods converge to the desired target, the ML-based approach achieves the solution at a significantly lower cost $J_\text{ctrl}$.
Define the Monte Carlo estimate of the control cost as
\begin{equation}
J_\text{ctrl} = \frac{1}{N_p} \sum_{p=1}^{N_p} \sum_{k=1}^N \|u_k(x_k^{(p)} ) \|^2 \Delta t_k,    
\end{equation} 
where $N_p$ is the number of particles.

\begin{figure}[htbp]
\centering
\includegraphics[width=0.4\columnwidth]{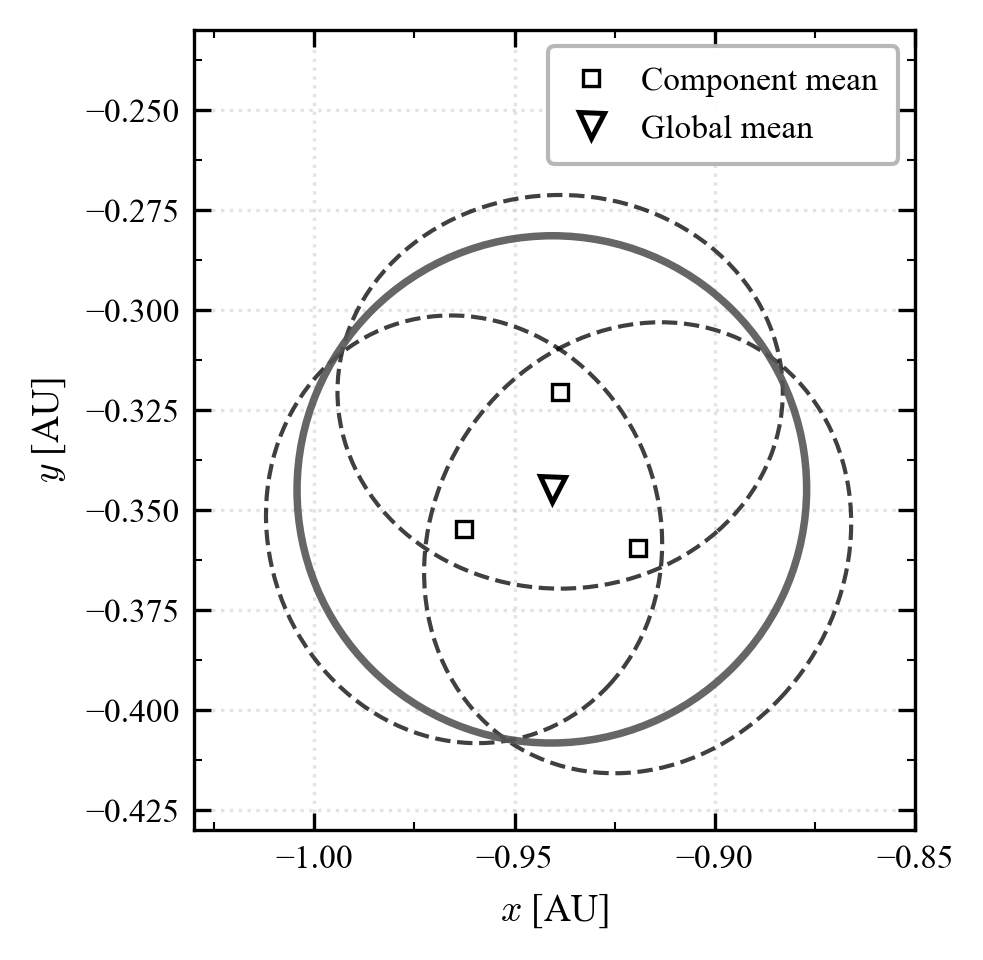}
\caption{Initial close-to-Gaussian distribution $\rho_0$. Covariance ellipses are computed at the $2\sigma$ level.} 
\label{fig7}
\end{figure}

\begin{figure}[htbp]
\centering
\includegraphics[width=0.7\columnwidth]{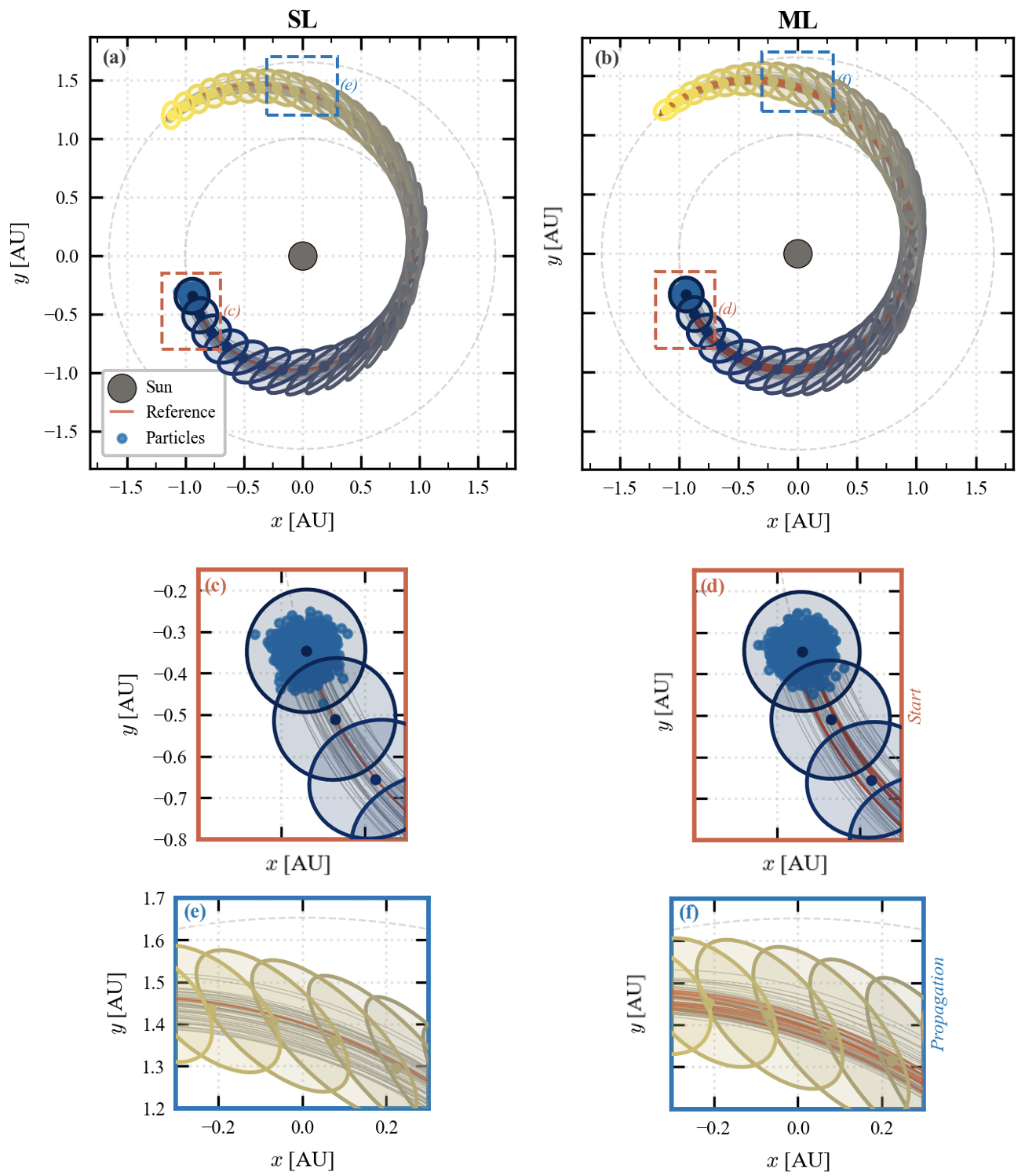}
\caption{Case II: position evolution under SL $(a)$ and ML $(b)$ policies. Layout as in Figure \ref{fig1}.}
\label{fig9}
\end{figure}
The ML approach achieves tighter distribution control throughout 
the transfer.

\begin{figure}[htbp]
\centering
\includegraphics[width= 0.7 \columnwidth]{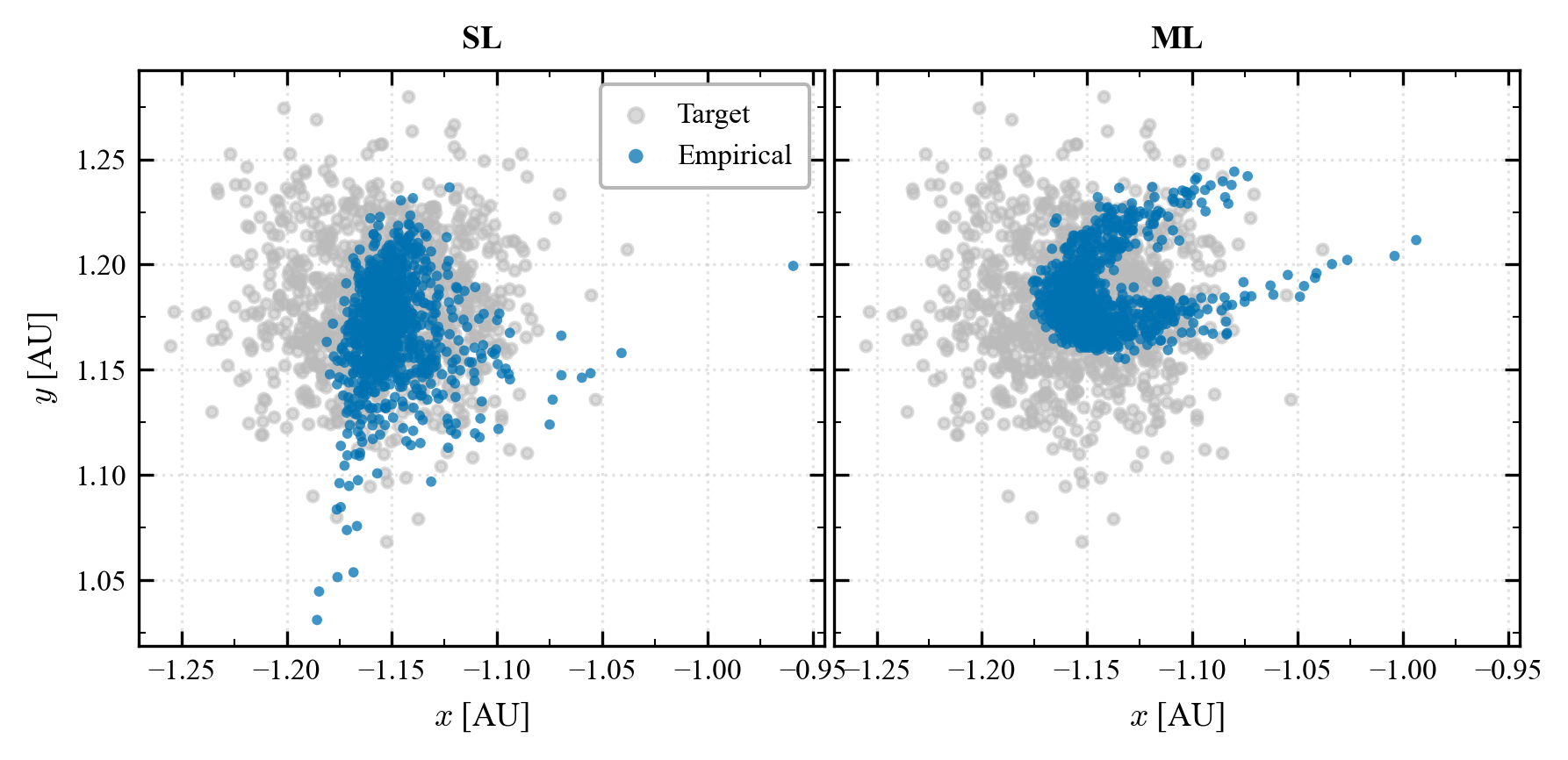}
\caption{Case II: terminal distribution matching.}
\label{fig6}
\end{figure}

\subsection{Metrics of Comparison}

The Earth-to-Mars interplanetary transfer introduces coupled 
non-linear gravity dynamics according to \eqref{dynamics_eq}, 
resulting in a state distribution that heavily strays from the 
initial Gaussian assumptions over the integration horizon. 
It is evident that the Single Linearization (SL) policy, which 
relies on an affine approximation centered around the aggregate 
global mean of the GMM, struggles to accurately capture the 
structural dispersion of the probability mass. To quantify 
terminal distribution matching, we report the Sliced 
Wasserstein-2 (Sliced W$_2$) distance, a tractable approximation of the Wasserstein-2 distance obtained by averaging one-dimensional 
optimal transport costs over random projections \cite{bonneel2015sliced}; 
a lower W$_2$ value indicates that the empirical terminal 
distribution of the Monte Carlo samples is closer to the desired 
target $\rho_T$. 

The generated feedback control commands from the SL approach
 manifest a $10.62\%$ (Case I) and $18.94\%$ (Case II) 
higher cost with respect to the ML case, while also yielding a 
larger Sliced W$_2$ distance, confirming that SL produces a less 
accurate terminal distribution match. The quantitative metrics 
obtained are summarized in Table~\ref{tab:metrics}.

\begin{table}[t]
\centering
\caption{Monte Carlo estimate of control cost comparison for Single vs Multiple Linearization. \textit{Note:} $\boldsymbol{J}_{\mathrm{ctrl}}$ is reported in units of $10^{-7}$, while Sliced $W_2$ is reported in units of $10^{-3}$. (I) and (II) denote Case I and Case II, respectively.}
\begin{tabular}{lccccc} 
\toprule
\textbf{Mode} & $\boldsymbol{J}_{\mathrm{ctrl}}$ (I) & $\boldsymbol{J}_{\mathrm{ctrl}}$ (II) & Sliced $W_2$ (I) & Sliced $W_2$ (II)\\
\midrule
SL & $1.82695$ & $1.95150$ & $4.44686$ & 13.5541\\
ML & $1.65157$ & $1.64075$ & $1.25304$ & 12.4214\\ 
\bottomrule
\end{tabular}
\label{tab:metrics}
\end{table}

\section{Conclusions}
This paper introduced a multiple-linearization approach to nonlinear stochastic density steering problems. 
By decomposing the problem into a mixture of local OCS subproblems, each linearizing the original problem around a different mean trajectory, the proposed method accurately handles multi-modal uncertainty in highly nonlinear regimes.
Theoretical analysis demonstrated that multiple linearizations provide 
tighter error bounds than single linearization when steering multi-modal 
distributions. 
Numerical results on the Earth-to-Mars transfer problem validated this finding, showing cost reduction compared to single linearization 
while maintaining terminal distribution constraints.
\vspace{-2mm}
\section*{Acknowledgment}
The authors express their thanks to Luigi Tesio for his assistance in designing Figure \ref{fig:description}.
George Rapakoulias acknowledges support from the A. Onassis Foundation.
\bibliography{refs}

\end{document}